\documentclass{lmcs}
\pdfoutput=1

\usepackage{silence}
\usepackage{lastpage}
\lmcsdoi{16}{1}{25}
\lmcsheading{}{\pageref{LastPage}}{}{}%
{Feb.~14,~2018}{Feb.~25,~2020}{}

\usepackage[utf8]{inputenc}
\usepackage{amssymb}
\usepackage{amsmath, amsthm, latexsym}
\usepackage{amsfonts}
\usepackage{xspace}
\usepackage{url}
\usepackage{ifthen}
\usepackage{theoremref}

\keywords{enumeration, first-order, constant delay, bounded expansion}

\newcommand{\CDlin}{\ensuremath{\textsc{\small CD}\!\circ\!\textsc{\small Lin}}\xspace}
\newcommand{\set}[1]{\ensuremath{\{#1\}}\xspace}

\newcommand{\cD}{\ensuremath{{\textbf{D}}}\xspace}

\newcommand{\C}{\ensuremath{\mathcal{C}}\xspace}
\newcommand{\Cnum}[1]{\ensuremath{\C_{#1}}\xspace}
\newcommand{\Cp}{\Cnum{p}}
\newcommand{\Cq}{\Cnum{q}}

\newcommand{\indegname}{\ensuremath{\Gamma_{\C}}\xspace}
\newcommand{\indeg}[1]{\ensuremath{\Gamma_{\C}(#1)}\xspace}

\newcommand{\sig}[1]{\ensuremath{\sigma_{\C}(#1)}\xspace}
\newcommand{\sigp}{\sig{p}}
\newcommand{\sigq}{\sig{q}}

\newcommand{\indegnameprim}{\ensuremath{\Gamma_{\C'}}\xspace}

\newcommand{\sigprim}[1]{\ensuremath{\sigma_{\C'}(#1)}\xspace}

\newcommand{\sigf}[1]{\ensuremath{\alpha_{\C}(#1)}\xspace}
\newcommand{\sigfp}{\sigf{p}}
\newcommand{\sigfq}{\sigf{q}}

\newcommand{\graph}{\ensuremath{\textbf{G}}\xspace}

\newcommand{\graphH}{\ensuremath{\textbf{H}}\xspace}
\newcommand{\cfun}[1]{\ensuremath{\textsc{f}\vec{\textbf{#1}}}\xspace}
\newcommand{\cvec}[1]{\ensuremath{\vec{\textbf{#1}}}\xspace}
\newcommand{\cG}{\cfun{G}}
\newcommand{\cgraph}{\cvec{\graph}}
\newcommand{\cGp}{\ensuremath{\cfun{G}'}\xspace}
\newcommand{\cGpp}{\ensuremath{\cfun{G}''}\xspace}
\newcommand{\cH}{\cfun{H}}

\newcommand{\FO}{\textup{FO}\xspace}
\newcommand{\MSO}{\textup{MSO}\xspace}

\newcommand{\DNF}{\textup{DNF}\xspace}

\newcommand{\Nat}{\mathbb{N}}
\newcommand{\Real}{\mathbb{R}}

\newcommand{\Aug}{\ensuremath{\cgraph_0 \subseteq \cgraph_{1} \subseteq \cgraph_2 \subseteq \ldots}}

\newcommand{\Dm}{\ensuremath{\Delta^{-}}}

\newcommand{\size}[1]{|\!|#1|\!|}

\newcommand{\adjacency}[1]{\ensuremath{\text{Adjacency}(#1)}}

\newcommand{\nil}{\ensuremath{\text{NULL}}\xspace}

\newcommand{\cw}[3]{\ensuremath{|#1(#2)|_{#3}}\xspace}
\newcommand{\psipp}{\ensuremath{\psi^{+}_{\textup{NF}}}\xspace}

\newcommand{\eA}{Example A-\!}
\newcommand{\eB}{Example B-\!}

\theoremstyle{definition}
\newtheorem{myexa}{\eA\!}
\newtheorem{myexb}{\eB\!}

\newcommand\todocolorempty[2]{
}
\newcommand\todocolornonempty[2]{
	\ifhmode\newline\fi
	\noindent\framebox{\parbox{\columnwidth}{\textcolor{#2}{#1}}}
}

\newcommand\wojtekempty[1]{\todocolorempty{{\bf Wojtek}: #1}{blue}}

\newcommand{\vb}[1]{}

\begin{document}

\title{First-order queries on classes of structures with bounded expansion}

\author[W.~Kazana]{Wojciech Kazana}
\author[L.~Segoufin]{Luc Segoufin}
\address{INRIA and ENS Cachan}

\begin{abstract}
  We consider the evaluation of first-order queries over classes of databases
  with \emph{bounded expansion}. The notion of bounded expansion is fairly broad
  and generalizes bounded degree, bounded treewidth and exclusion of at least
  one minor. It was known that over a class of databases with bounded
  expansion, first-order sentences could be evaluated in time linear in the
  size of the database. We give a different proof of this
  result. Moreover, we show that answers to first-order queries can be
  enumerated with constant delay after a linear time preprocessing.  We also
  show that counting the number of answers to a query can be done in time
  linear in the size of the database.
\end{abstract}

\maketitle

\section{Introduction}
Query evaluation is certainly the most important problem in databases. Given a
query $q$ and a database \cD it computes the set $q(\cD)$ of all tuples in the
output of $q$ on \cD. However, the set $q(\cD)$ may be larger than the database
itself as it can have a size of the form $n^l$ where $n$ is the size of the
database and $l$ the arity of the query. Therefore, computing entirely $q(\cD)$
may require too many of the available resources.

There are many solutions to overcome this problem. For instance one could
imagine that a small subset of $q(\cD)$ can be quickly computed and that this
subset will be enough for the user needs. Typically one could imagine computing the
top-$\ell$ most relevant answers relative to some ranking function or to provide a
sampling of $q(\cD)$ relative to some distribution. One could also imagine
computing only the number of solutions $|q(\cD)|$ or providing an efficient test
for whether a given tuple belongs to $q(\cD)$ or not.

In this paper we consider a scenario consisting in enumerating $q(\cD)$ with
constant delay. Intuitively, this means that there is a two-phase algorithm
working as follows: a preprocessing phase that works in time linear in the size
of the database, followed by an enumeration phase outputting one by one all the
elements of $q(\cD)$ with a constant delay between any two consecutive
outputs. In particular, the first answer is output after a time linear in the
size of the database and once the enumeration starts a new answer is being output
regularly at a speed independent from the size of the database. Altogether, the set
$q(\cD)$ is entirely computed in time $f(q)(n+|q(\cD)|)$ for some function $f$
depending only on $q$ and not on \cD.

One could also view a constant delay enumeration algorithm as follows. The preprocessing
phase computes in linear time an index structure representing the set $q(\cD)$ in
a compact way (of size linear in $n$). The enumeration algorithm is then a
streaming decompression algorithm.

One could also require that the enumeration phase outputs the answers in some
given order. Here we will consider the lexicographical order based on a
linear order on the domain of the database.

There are many problems related to enumeration. The main one is the model
checking problem. This is the case when the query is boolean, i.e., outputs only
\emph{true} or \emph{false}. In this case a constant delay enumeration algorithm is a Fixed Parameter Linear
(FPL) algorithm for the model checking problem of $q$, i.e., it works in time
$f(q)n$. This is a rather strong constraint as even the model checking
problem for conjunctive queries is not FPL (assuming some hypothesis in
parametrized complexity)~\cite{PY99}. Hence, in order to obtain constant delay
enumeration algorithms, we need to make restrictions on the queries and/or on
the databases. Here we consider first-order (\FO) queries over classes of structures
having ``bounded expansion''.

The notion of class of graphs with bounded expansion was introduced by Ne\v{s}et\v{r}il and Ossona de Mendez in~\cite{NesetrilMendez08}. Its precise
definition can be found in Section~\ref{section-augs}. At this point it is only useful
to know that it contains the class of graphs of bounded degree, the class of
graphs of bounded treewidth, the class of planar graphs, and any class of graphs
excluding at least one minor. This notion is generalized to classes of structures
via their Gaifman graphs or adjacency graphs.

For the class of structures with bounded degree and \FO queries
the model checking problem is in FPL~\cite{Seese96} and there also are constant delay
enumeration algorithms~\cite{DurandGrandjean07,WL11}. In the case of structures
of bounded treewidth and \FO queries (actually even \MSO queries with
first-order free variables) the model checking problem is also in FPL~\cite{Courcelle90}
and there are constant delay enumeration
algorithms~\cite{Bagan06,WL12}. For classes of structures with bounded
expansion the model checking problem for \FO queries was recently shown to be
in FPL~\cite{DKT11,GK11}.

\medskip

{\bf Our results can be summarized as follows}. For \FO queries and any class of
structures with bounded expansion:
\begin{itemize}
    \item we provide a new proof that the model checking problem can be solved in FPL,
    \item we show that the set of solutions to a query can be enumerated with constant delay,
    \item we show that computing the number of solutions can be done in FPL,
    \item we show that, after a preprocessing in time linear in the size of the database, one can test on input $\bar a$ whether $\bar a \in q(\cD)$ in constant time.
\end{itemize}

\medskip
\noindent
Concerning model checking, our method uses a different technique than the
previous ones. There are several characterizations of classes having bounded
expansion~\cite{NesetrilMendez08}. Among them we find characterizations via
``low tree depth coloring'' and ``transitive fraternal augmentations''. The
previous methods were based on the low tree depth coloring characterization
while ours is based on transitive fraternal augmentations. We show that it is
enough to consider quantifier-free queries in a given normal form. The normal
form is at the core of our algorithms for constant delay enumeration and for
counting the number of solutions. As for the previous proofs, we exhibit a
quantifier elimination method, also based on our normal form. Our quantifier
elimination method results in a quantifier-free query but over a recoloring
of a functional representation of a ``fraternal and transitive augmentation''
of the initial structure.

Our other algorithms (constant delay enumeration, counting the number of
solution or testing whether a tuple is a solution or not) start by eliminating
the quantifiers as for the model checking algorithm. The quantifier-free case
is already non trivial and require the design and the computation of new index
structures. For instance consider the simple query $R(x,y)$. Given a pair
$(a,b)$ we would like to test whether $(a,b)$ is a tuple of the database in
constant time. In general, index structures can do this with $\log n$ time. We
will see that we can do constant time, assuming bounded expansion.

In the presence of a linear order on the domain of the database, our constant
delay algorithm can output the answers in the corresponding lexicographical
order.

\paragraph{Related work}

We make use of a functional representation of the initial structures. Without
this functional representations we would not be able to eliminate all
quantifiers. Indeed, with this functional representation we can talk of a node
at distance 2 from $x$ using the quantifier-free term $f(f(x))$, avoiding the
existential quantification of the middle point. This idea was already taken
in~\cite{DurandGrandjean07} for eliminating first-order quantifiers over structures of
bounded degree. Our approach differs from theirs in the fact that in the bounded
degree case the functions can be assumed to be permutations (in particular
they are invertible) while this is no longer true in our setting, complicating
significantly the combinatorics.
\wojtekempty{Should we mention that~\cite{DKT11}
also use functional representation? They do slightly different thing:
they have full/mixed signatures and with functions they encode the
closure of the low tree-depth embedding (similarly to~\cite{GK11}) and
not the actual graph. For now we do not comment on this. As for~\cite{GK11} we do not comment either as some/many details are not
present in the paper.}

Once we have a quantifier-free query, constant delay enumeration could also
be obtained using the characterization of bounded expansion based on low tree
depth colorings. Indeed, using this characterization one can easily show that
enumerating a quantifier-free query over structures of bounded expansion
amounts in enumerating an \MSO query over structures of bounded tree-width and
for those known algorithms exist~\cite{Bagan06,WL12}. However, the known
enumeration algorithms of \MSO over structures of bounded treewidth are rather
complicated while our direct approach is fairly simple. Actually, our proof
shows that constant delay enumeration of \FO queries over structures of
bounded treewidth can be done using simpler algorithms than for \MSO queries. Moreover,
it gives a constant delay algorithm outputting the solutions in lexicographical
order. No such algorithms were known for \FO queries over structures of bounded
treewidth. In the bounded degree case, both enumeration algorithms
of~\cite{DurandGrandjean07,WL11} output their solutions in lexicographical
order.

Similarly, counting the number of solutions of a quantifier-free query over
structures of bounded expansion reduces to counting the number of solutions of
a \MSO query over structures of bounded treewidth. This latter problem is known to be in
FPL~\cite{ALS91}. We give here a direct and simple proof of this fact for \FO
queries over structures of bounded expansion.

Our main result is about enumeration of first-order queries. We make use of a
quantifier elimination method reducing the general case to the quantifier-free
case. As a special we obtain a new proof of the linear time model checking
algorithm, already obtained in~\cite{DKT11,GK11}. Both these results were also
obtained using (implicitly or explicitly) a quantifier elimination method. As
our enumeration of quantifier-free query also needs a specific normal form,
we could not reuse the results of~\cite{DKT11,GK11}. Hence we came up with our
own method which differ in the technical details if not in the main ideas.

In~\cite{DKT11} it is also claimed that the index structure used for quantifier
elimination can be updated in constant time. It then tempting to think that
enumeration could be achieved by adding each newly derived output tuple to the index
structure and obtain the next output in constant time using the updated index.
This idea does not work because constant update time can only be achieved if
the inserted tuple does not modify too much the structure of the underlying graph. In
particular the new structure must stay within the class under
investigation. This is typically not the case with first-order query that may
produce very dense results.

\medskip

This paper is the journal version of~\cite{DBLP:conf/pods/KazanaS13}. All proofs
are now detailed and the whole story has been simplified a bit, without changing
the key ideas. Since the publication of the conference version, constant delay
enumeration has been obtained for first-order queries over any class of
structures having local bounded expansion~\cite{DBLP:conf/icdt/SegoufinV17} or
being nowhere dense~\cite{DBLP:conf/pods/SchweikardtSV18}
These two classes of structures generalize bounded expansion. However the
preprocessing time that has been achieved for these two classes is not linear but
pseudo-linear (i.e for any $\epsilon$ there is an algorithm working in time
$O(n^{1+\epsilon})$) and the enumeration algorithms are significantly more complicated.

\section{Preliminaries}
In this paper a database is a finite relational structure.  A \emph{relational
  signature} is a tuple $\sigma=(R_{1}, \ldots, R_{l})$, each $R_i$ being a
relation symbol of arity $r_i$. A \emph{relational structure} over $\sigma$ is
a tuple $\cD = \left(D, R^{\cD}_{1}, \ldots, R^{\cD}_{l} \right)$, where $D$ is
\emph{the domain} of \cD and $R^{\cD}_{i}$ is a subset of $D^{r_{i}}$.
We will often write $R_i$ instead of $R_i^{\cD}$ when $\cD$ is clear from the
context.

We use a standard notion of size. The \emph{size} of $R_i^{\cD}$, denoted
$\size{R_i^\cD}$ is the number of tuples in $R_i^{\cD}$ multiplied by the arity
$r_i$. The \emph{size} of the domain of $\cD$, denoted $|\cD|$, is the number
of elements in $D$. Finally the \emph{size} of \cD, denoted by $\size{\cD}$, is
\begin{equation*}
\size{\cD}=|\cD| + \Sigma_{R_i \in \sigma} \size{R_i^{\cD}}.
\end{equation*}

By \emph{query} we mean a formula of first-order logic, \FO, built
from atomic formulas of the form $x=y$ or $R_{i}(x_{1}, \ldots, x_{r_{i}})$ for
some relation $R_{i}$, and closed under the usual Boolean connectives
($\neg,\vee,\wedge$) and existential and universal quantifications
($\exists,\forall$). We write $\phi(\bar x)$ to denote a query whose free
variables are $\bar x$, and the number of free variables is called the
\emph{arity of the query}. A \emph{sentence} is a query of arity 0.  We use the
usual semantics, denoted $\models$, for first-order.  Given a structure \cD and
a query $q$, an \emph{answer} to $q$ in \cD is a tuple $\bar a$ of elements
of \cD such that $\cD \models q(\bar a)$. We write $q(\cD)$ for the set
of answers to $q$ in \cD, i.e. $q(\cD)=\set{\bar a ~|~ \cD \models q(\bar{a})}$. As usual, $|q|$ denotes the size of $q$.

Let \C be a class of structures. The model checking problem for \FO over \C is
the computational problem of given first-order \emph{sentence} $q$ and a database
$\cD\in\C$ to test whether $\cD\models q$ or not.

We now introduce our running examples.

\begin{myexa}\label{ex-a-mc}
  The first query has arity 2 and returns pairs of nodes at distance 2 in a
  graph. The query is of the form $\exists z E(x,z) \wedge
  E(z,y)$.

  Testing the existence of a solution to this query can be easily done in time
  linear in the size of the database. For instance one can go trough all nodes
  of the database and check whether it has non-nill in-degree and out-degree. The
  degree of each node can be computed in linear time by going through all edges
  of the database and incrementing the counters associated to its endpoints.
\end{myexa}

\begin{myexb}\label{ex-b-mc}
The second query has arity 3 and returns triples $(x,y,z)$ such that $y$ is
connected to $x$ and $z$ via an edge but $x$ is not connected to $z$. The query
is of the form $E(x,y) \wedge E(y,z) \wedge \lnot E(x,z)$.

It is not clear at all how to test the existence of a solution to this query in
time linear in the size of the database. The problem is similar to the one of
finding a triangle in a graph, for which the best know algorithm has complexity
even slightly worse than matrix multiplication~\cite{AlonYZ95}.
If the degree of the input structure is bounded by a constant $d$, we can test the
existence of a solution in linear time by the following algorithm. We first go
through all edges $(x,y)$ of the database and add $y$ to a list associated to
$x$ and $x$ to a list associated to $y$. It remains now to go through all
nodes $y$ of the database, consider all pairs $(x,z)$ of nodes in the
associated list (the number of such pairs is bounded by $d^2$) and then test
whether there is an edge between $x$ and $z$ (by testing whether $x$ is in the
list associated to $z$).

We aim at generalizing this kind of reasoning  to structures with bounded expansion.
\end{myexb}

Given a query $q$, we care about ``enumerating'' $q(\cD)$ efficiently. Let \C
be a class of structures. For a query $q(\bar x)$, \emph{the enumeration problem of
$q$ over \C} is, given a database $\cD\in\C$, to output the elements of $q(\cD)$
one by one with no repetition. The maximum time between any two consecutive
outputs of elements of $q(\cD)$ is called \emph{the delay}. The
definition  below requires a constant time delay.
We formalize these notions in the forthcoming section.

\subsection{Model of computation and enumeration}
We use Random Access Machines (RAM) with addition and uniform cost measure as a
model of computation. For further details on this model and its use in logic
see~\cite{DurandGrandjean07}.  In the sequel we assume that the input
relational structure comes with a linear order on the domain. If not, we use the
one induced by the encoding of the database as a word. Whenever we iterate
through all nodes of the domain, the iteration is with respect to
the initial linear order.

We say that the enumeration problem of $q$ over a class \C of structures is in
the class \CDlin, or equivalently that we can enumerate $q$ over \C with constant
delay, if it can be solved by a RAM algorithm which, on input $\cD\in\C$, can
be decomposed into two phases:
\begin{itemize} \itemsep1pt \parskip0pt \parsep0pt
	\item a precomputation phase that is performed in time $O(\size{\cD})$,
	\item an enumeration phase that outputs $q(\cD)$ with no repetition and
          a constant delay between two consecutive outputs. The enumeration
          phase has full access to the output of the precomputation phase but
          can use only a constant total amount of extra memory.
\end{itemize}

\noindent
Notice that if we can enumerate $q$ with constant delay, then all answers
can be output in time $O(\size{\cD}+|q(\cD)|)$ and the first output is computed in
time linear in $\size{\cD}$. In the particular case of boolean queries, the
associated model checking problem must be solvable in time linear in $\size{\cD}$.
Notice also that the total amount of memory used after computing all answers is
linear in $\size{\cD}$, while a less restrictive definition requiring only a
constant time delay between any two outputs may yield in a total amount of
memory linear in $\size{\cD} +\size{q(\cD)}$.

Note that we measure the running time complexity as a function of
$\size{\cD}$. The multiplicative factor will depend on the class \C of database
under consideration and, more importantly, on the query $q$. In our case we
will see that the multiplicative factor is non elementary in $|q|$ and that
cannot be avoided, see the discussion in the conclusion section.

We may in addition require that the enumeration phase outputs the answers to
$q$ using the lexicographical order. We then say that we can enumerate $q$ over
\C with constant delay in lexicographical order.

\begin{myexa}\label{ex-a-enum}
Over the class of all graphs, we cannot enumerate pairs of nodes at distance 2
with constant delay unless the  Boolean Matrix Multiplication problem can be
solved in quadratic time~\cite{BDG07}. However, over the class of graphs of degree $d$, there is
a simple constant delay enumeration algorithm. During the preprocessing phase, we
associate to each node the list of all its neighbors at distance 2. This can be done in
time linear in the size of the database as in \eB~\ref{ex-b-mc}. We then color in blue all
nodes having a non empty list and make sure each blue node points to the next
blue node (according to the linear order on the domain). This also can be
done in time linear in the size of the database and concludes the preprocessing
phase. The enumeration phase now goes through all blue nodes $x$ using the
pointer structure and, for each of them, outputs all pairs $(x,y)$ where $y$ is
in the list associated to $x$.
\end{myexa}

\begin{myexb}\label{ex-b-enum}
  Over the class of all graphs, the query of this example cannot be enumerated with
  constant delay because, as mentioned in \eB~\ref{ex-b-mc}, testing whether there is one solution is
  already non linear. Over the class of graphs of bounded degree, there is a
  simple constant delay enumeration algorithm, similar to the one from \eA~\ref{ex-a-enum}.
\end{myexb}

Note that in general constant delay enumeration algorithms are not closed under
any boolean operations. For instance if $q$ and
$q'$ can be enumerated with constant delay, we cannot necessarily enumerate $q \lor q'$ with constant delay
as enumerating one query after the other would break the ``no repetition''
requirement.  However, if we can enumerate with constant delay in the
lexicographical order, then a simple argument that resembles the problem of
merging two sorted lists shows closure under union:

\begin{lem}\label{disjunct-enum}
  If both queries $q(\bar x)$ and $q'(\bar x)$ can be enumerated in
  lexicographical order with constant delay then the same is true for
  $q(\bar x) \lor q'(\bar x)$.
\end{lem}
\begin{proof}
The preprocessing phase consists in the preprocessing phases of the enumeration
algorithms for $q$ and $q'$.

The enumeration phase keeps two values, the smallest element from $q(\cD)$ that was not yet output
and similarly the smallest element from $q'(\cD)$ that was not yet output. It then outputs the smaller of
the two values and replaces it in constant time with the next element from the appropriate set using the
associated enumeration procedure. In case the elements are equal, the value is output
once and both stored values are replaced with their appropriate successors.
\end{proof}

It will follow from our results that the enumeration problem of \FO over the
class of structures with ``bounded expansion'' is in \CDlin. The notion of
bounded expansion was defined in~\cite{NesetrilMendez08} for graphs and then
it was generalized to structures via their Gaifman or Adjacency graphs.
We start with defining it for graphs.

\subsection{Graphs with bounded expansion and augmentation}\label{section-augs}

By default a graph has no orientation on its edges and has colors on its
vertices.  In an \emph{oriented} graph every edge is an arrow going from the
source vertex to its target. We can view a (oriented or not) graph as a
relational structure
$\graph=(V^\graph,E^\graph,P_1^\graph, \ldots, P_l^\graph)$, where $V^\graph$
is the set of nodes, $E^\graph \subseteq V^{2}$ is the set of edges and, for
each $1 \leq i \leq l$, $P^\graph_i$ is a predicate of arity $1$, i.e., a color. We omit the
subscripts when \graph is clear from the context. In the nonoriented case, $E$
is symmetric and irreflexive and we denote by $\set{u,v}$ the edge between $u$
and~$v$. In the oriented case we denote by $(u,v)$ the edge from $u$ to~$v$. We
will use the notation \cgraph when the graph is oriented and $\graph$ in the
nonoriented case.  An \emph{orientation} of a graph \graph is any graph
$\cvec\graphH$ such that $\set{u,v}\in E^\graph$ implies
$(u,v)\in E^{\cvec\graphH}$ or $(v,u)\in E^{\cvec\graphH}$.  The
\emph{in-degree} of a node $v$ of \cgraph is the number of nodes $u$ such that
$(u,v) \in E$.  We denote by $\Dm(\cgraph)$ the maximum in-degree of a node of
\cgraph. Among all orientations of a graph \graph, we choose the following one,
which is computable in time linear in $\size{\graph}$. It is based on the
degeneracy order of the graph. We find the first node of minimal degree, orient
its edges towards it and repeat this inductively in the induced subgraph
obtained by removing this node. The resulting graph, denoted $\cgraph_0$, has
maximum in-degree which is at most twice the optimal value and that is enough
for our needs.

In~\cite{NesetrilMendez08} several equivalent definitions of bounded expansion
were shown. We present here only the one we will use, exploiting the notion of
``augmentations''.

Let \cgraph be an oriented graph. A \emph{$1$-transitive fraternal augmentation of
  \cgraph} is any graph $\cvec\graphH$ with the same vertex set as \cgraph and the
same colors of vertices, including all edges of \cgraph (with their
orientation) and such that for any three vertices $x, y, z$ of \cgraph we have
the following:
\begin{description} \itemsep1pt \parskip0pt \parsep0pt
\item[(transitivity)] if $(x, y)$ and $(y, z)$ are edges in \cgraph,
	then $(x, z)$ is an edge in \cvec\graphH,

\item[(fraternity)] if $(x, z)$ and $(y, z)$ are edges in \cgraph,
	then at least one of the edges: $(x,y)$, $(y,x)$ is in \cvec\graphH,

      \item[(strictness)] moreover, if \cvec\graphH contains an edge that was
        not present in \cgraph, then it must have been added by one of the
        previous two rules.
\end{description}

\noindent
Note that the notion of $1$-transitive fraternal augmentation is not a
deterministic operation.  Although transitivity induces precise edges,
fraternity implies nondeterminism and thus there can possibly be many different
$1$-transitive fraternal augmentations. We care here about choosing the
orientations of the edges resulting from the fraternity rule in order to
minimize the maximum in-degree.

Following~\cite{NesetrilMendez08partII} we fix a deterministic algorithm
computing a ``good'' choice of orientations of the edges induced by the
fraternity property. The precise definition of the algorithm is not important
for us, it only matters here that the algorithm runs in time linear in the
size of the input graph (see Lemma~\ref{lemma-frat-aug} below).
With this algorithm fixed, we can now speak of {\bf the}
$1$-transitive fraternal augmentation of \cgraph.

Let $\cgraph_0$ be an oriented graph. The \emph{transitive fraternal augmentation}
of $\cgraph_0$ is the sequence $\Aug$ such that for each $i \geq 1$
the graph $\cgraph_{i+1}$ is the $1$-transitive fraternal augmentation of
$\cgraph_{i}$. We will say that $\cgraph_i$ is the $i$-th augmentation of $\cgraph_0$.
Similarly we denote the \emph{transitive fraternal augmentation}
of a nonoriented graph $\graph$ by considering the orientation $\cgraph_0$
based on the degeneracy order as explained above.

\begin{defiC}[\cite{NesetrilMendez08}]\label{aug-bexp}
Let \C be a class of graphs. \C has bounded expansion if
there exists a function $\indegname: \Nat \rightarrow \Real$ such that for each
graph $\graph \in \C$  its transitive fraternal augmentation
$\Aug$ of $\graph$ is such that for each $i \geq 0$ we have
$\Dm(\cgraph_{i}) \leq \indeg{i}$.
\end{defiC}

Consider for instance a graph of degree $d$. Notice that the $1$-transitive
fraternal augmentation introduces an edge between nodes that were at
distance at most 2 in the initial graph. Hence, when starting with a graph of
degree $d$, we end up with a graph of degree at most $d^2$. This observation
shows that the class of graphs of degree $d$ has bounded expansion as
witnessed by the function $\Gamma(i)=d^{2^i}$. Exhibiting the function $\Gamma$
for the other examples of classes with bounded expansion mentioned in the introduction:
bounded treewidth, planar graphs, graphs excluding at least one minor, requires more
work~\cite{NesetrilMendez08}.

The following lemma shows that within a class \C of bounded expansion the
$i$-th augmentation of  $\graph\in\C$ can be computed in linear time, the
linear factor depending on $i$ and on $\C$.

\begin{lemC}[\cite{NesetrilMendez08partII}]\label{lemma-frat-aug}
  Let \C be a class of bounded expansion. For each $\graph\in\C$ and each
  $i\geq 0$,
  $\cgraph_i$ is computable from $\graph$ in time
  $O(\size{\graph})$.
\end{lemC}

A transitive fraternal augmentation introduces new edges in the graphs in
a controlled way. We will see that we can use these extra edges in order to
eliminate quantifiers in a first-order query. Lemma~\ref{lemma-frat-aug} shows
that this quantifier elimination is harmless for enumeration as it can be done in time linear in the size of the
database and can therefore be done during the preprocessing phase.

\subsection{Graphs of bounded in-degree as functional structures}\label{subsec-func-defs}

Given the definition of bounded expansion it is convenient to work with
oriented graphs. These graphs will always be such that the maximum in-degree is
bounded by some constant depending on the class of graphs under investigation.
It is therefore convenient for us to represent our graphs as functional
structures where the functions links the current node with its predecessors.
This functional representation turns out to be also useful for eliminating
some quantifiers.

A \emph{functional signature} is a tuple $\sigma=(f_{1}, \ldots,
f_{l},P_1,\ldots,P_m)$, each $f_i$ being a functional symbol of arity $1$ and
each $P_i$ being an unary predicate. A \emph{functional structure} over
$\sigma$ is then defined as for relational structures.  \FO is defined as usual
over the functional signature. In particular, it can use atoms of the form
$f(f(f(x)))$, which is crucial for the quantifier elimination step of
Section~\ref{section-MC} as the relational representation would require
existential quantification for denoting the same element.  A graph $\cgraph$ of
in-degree $l$ and colored with $m$ colors can be represented as a functional
structure $\cG$, where the unary predicates encode the various colors and
$v=f_i(u)$ if $v$ is the $i^{\text{th}}$ element (according to some
arbitrary order that will not be relevant in the sequel) such that $(v,u)$ is
an edge of $\cgraph$.
We call such node $v$ the \emph{$i^{\text{th}}$ predecessor} of $u$ (where
``$i^{\text{th}}$ predecessor'' should really be viewed as an abbreviation for
``the node $v$ such that $f_i(u)=v$'' and not as a reference to the chosen
order). If we do not care about the $i$ and we only want to say that $v$ is the
image of $u$ under some function, we call it a \emph{predecessor} of $u$.
If a node $u$ has $j$ predecessors with $j<l$, then we set $f_k(u)=u$ for all
$k>j$. This will allow us to detect how many predecessors a given node has
without quantifiers by checking whether $f_j(u)=u$ or not.
Given a nonoriented graph $\graph$ we define $\cG$ to be the functional representation of
$\cgraph_0$ as described above.  Note that $\cG$ is computable in time linear in
$\size{\graph}$ and that for each first order query $\phi(\bar x)$, over the
relational signature of graphs, one can easily compute a first order query $\psi(\bar x)$,
over the functional signature, such that
$\phi(\graph)=\psi(\cG)$.

\begin{myexa}\label{ex-a-enum-f}
Consider again the query computing nodes at distance 2 in a nonoriented
graph. There are four possible ways to orient a path of length~2. With the
functional point of view we further need to consider all possible
predecessors. Altogether the distance 2 query is now equivalent to:
\begin{equation*}
\bigvee_{f,g}
\begin{aligned}
f(g(x)) = y ~\lor~ g(f(y)) = x ~\lor~ f(x) = g(y) ~\lor
 \exists z~~ f(z) = x \wedge g(z) = y
\end{aligned}
\end{equation*}
where the small disjuncts correspond to the four possible orientations and the
big one to all possible predecessors, each of them corresponding to a function
name, whose number depends on the function $\indegname$.
\end{myexa}

\begin{myexb}\label{ex-b-enum-f}
Similarly, the reader can verify that the query of \eB~\ref{ex-b-mc} is equivalent to:
\begin{align*}
\bigvee_{f,g}
~~\bigwedge_{h}~~ (h(x) \neq z \wedge h(z) \neq x)
\wedge~~ &\big[(f(x) = y \wedge g(y) = z)\\[-.3cm]
&\vee (x = f(y) \wedge g(y) = z)\\
&\vee (f(x) = y \wedge y = g(z))\\
&\vee (x = f(y) \wedge y = g(z))\big].
      \end{align*}
\end{myexb}

{\bf Augmentation for graphs as functional structures.}
The notion of $1$-transitive fraternal augmentation can be adapted directly to
the functional setting. However it will be useful for us to enhance it with
extra information. In particular it will be useful to remember at which stage the
extra edges are inserted. We do this as follows.

Given a graph $\cG$, its $1$-transitive fraternal augmentation $\cG'$ is constructed as follows.
The signature of $\cG'$ extends the signature of $\cG$ with new function symbols for
taking care of the new edges created during the expansion and $\cG'$ is then an
expansion, in the logical sense, of $\cG$ over this new signature: $\cG$ and
$\cG'$ agree on the relations in the old signature.

For any pair of functions $f$ and $g$ in the signature of $\cG$ there is a new
function $h$ in the signature of $\cG'$ representing the transitive part of the
augmentation. It is defined as the composition of $f$ and $g$, i.e. $h^{\cG'}=f^{\cG} \circ g^{\cG}$

Similarly, for any pair of functions $f$ and $g$ in the signature of $\cG$, and
any node $x$ in the domain of both $f^{\cG}$ and $g^{\cG}$ there will be a
function $h$ in the new signature representing the fraternity part of the
augmentation. I.e $h$ is such that $h^{\cG'}(f^{\cG'}(x))=g^{\cG}(x)$ or
$h^{\cG'}(g^{\cG'}(x))=f^{\cG'}(x)$.

Given a class $\C$ of bounded expansion, the guarantees that the number of new function symbols
needed for the $i$-th augmentation is bounded by $\indeg{i}$ and does not
depend on the graph. Hence a class \C of bounded expansion generates finite functional signatures
$\sig{0} \subseteq \sig{1} \subseteq  \sig{2} \subseteq \ldots$ such that for any graph $\graph\in \C$ and for all $i$:
\begin{enumerate} \itemsep1pt \parskip0pt \parsep0pt
\item\label{aaa} $\cG_i$ is a functional structure over $\sig{i}$ computable in
  linear time from $\graph$,
\item\label{bbb} $\cG_{i+1}$ is an  expansion of $\cG_{i}$,
	\item\label{queryholds} for every \FO query $\phi(\bar x)$ over $\sig{i}$ and every $j \geq i$
	we have that $\phi(\cG_{i}) = \phi(\cG_{j})$.
\end{enumerate}

\noindent
We denote by $\sigf{i}$ the number of function symbols of $\sig{i}$.  Notice
that we have $\sigf{i} \leq \Sigma_{j\leq i} \indeg{j}$.

\medskip

We say that a functional signature $\sigma'$ is a \emph{recoloring} of $\sigma$
if it extends $\sigma$ with some extra unary predicates, also denoted as
\emph{colors}, while the functional part remains intact. Similarly, a
functional structure $\cG'$ over $\sigma'$ is a \emph{recoloring} of $\cG$ over
$\sigma$ if $\sigma'$ is a recoloring of $\sigma$ and $\cG'$ differs from $\cG$
only in the colors in $\sigma'$. We write \emph{$\phi$ is over a recoloring of
  $\sigma$} if $\phi$ is over $\sigma'$ and $\sigma'$ is a recoloring of
$\sigma$. Notice that the definition of bounded expansion is not sensitive to
the colors as it depends only on the binary predicates, hence adding any fixed finite
number of colors is harmless.

Given a class \C of graphs, for each $p \geq 0$, we define $\Cp$ to be the class of all recolorings
$\cG'_{p}$ of $\cG_p$ for some $\graph \in \C$. In other words $\Cp$ is the class of functional representations of all
recolorings of all $p$-th augmentations of graphs from $\C$. Note that all
graphs from $\Cp$ are recolorings of a structure in $\sig{p}$, hence they use at
most $\sigf{p}$ function symbols.

From now on we assume that all graphs from \C and all queries are
in their functional representation. It follows from the discussion above that this is without loss of generality.

\subsection{From structures to graphs}\label{sec-str-gr}

A class of structures is said to have bounded expansion if the set of adjacency
graphs of the structures of the class has bounded expansion.

The \emph{adjacency graph} of a relational structure $\cD$, denoted by
$\adjacency{\cD}$, is a functional structure defined as follows. The set of
vertices of $\adjacency{\cD}$ is $D \cup T$ where $D$ is the domain of $\cD$
and $T$ is the set of tuples occurring in some relation of \cD. For each
relation $R_i$ in the schema of \cD, there is a unary symbol $P_{R_i}$ coloring
the elements of $T$ belonging to $R_i$. For each tuple $t=(a_1,\dots,a_{r_i})$
such that $\cD \models R_i(t)$ for some relation $R_i$ of arity $r_i$, we have
an edge $f_j(t) = a_j$ for all $j\leq r_i$.

\begin{obs}\label{obs-str-adj}
It is immediate to see that for every relational structure $\cD$ we can
compute $\adjacency{\cD}$ in time $O(\size{\cD})$.
\end{obs}

Let \C be a class of relational structures. We say that \C has \emph{bounded
  expansion} if the class \C' of adjacency graphs (seen as graphs)
of structures from \C has bounded expansion.

\begin{rem}
  In the literature, for instance~\cite{DKT11,GK11}, a class \C of
  relational structures is said to have bounded expansion if the class of their
  Gaifman graphs has bounded expansion. It is easy to show that if the class of
  \emph{Gaifman graphs} of structures from \C has bounded expansion then the
  class of \emph{adjacency graphs} of structures from \C also has bounded
  expansion. The converse is not true in general. However the converse holds if
  the schema is fixed, i.e. \C is a class of structures all having the same schema.
  We refer to~\cite{thesis-kazana} for the simple proofs of these facts.
\end{rem}

Let $\indegnameprim$ be the function given by Definition~\ref{aug-bexp} for
\C'. The following lemma is immediate. For instance $R(\bar x)$ is rewritten as
$\exists y P_R(y) \wedge \bigwedge_{1 \leq i \leq r} f_i(y) = x_i$.

\begin{lem}\label{lemma-structures-to-graphs}
  Let \C be a class of relational structures with bounded expansion and let \C'
  be the underlying class of adjacency graphs.  Let $\phi(\bar x) \in \FO$. In
  time linear in the size of $\phi$ we can find a query $\psi(\bar x)$ over
  $\sigprim{0}$ such that for all $\cD \in \C$ we have $\phi(\cD) =
  \psi(\adjacency{\cD})$.
\end{lem}
As a consequence of Lemma~\ref{lemma-structures-to-graphs} it follows that
model checking, enumeration and counting of first-order queries over relational
structures reduce to the graph case. Therefore in the rest of the paper we will
only concentrate on the graph case (viewed as a functional structure), but the
reader should keep in mind that all the results stated over graphs extend to
relational structures via this lemma.

\section{Normal form for quantifier-free first-order queries}

We prove in this section a normal form on quantifier-free first-order
formulas. This normal form will be the ground for all our algorithms later on.
It says that, modulo performing some extra augmentation steps, a
quantifier-free formula has a very simple form.

Fix class \C of graphs with bounded expansion. Recall that we are now
implicitly assuming that graphs are represented as functional
structures.

A formula is \emph{simple} if it does not contain atoms of the form
$f(g(x))$, i.e., it does not contain any compositions of functions.
We first observe that, modulo augmentations, any formula can be
transformed into a simple one.

\begin{lem}\label{lemma-simple}
  Let $\psi(\bar x)$ be a formula over a recoloring of $\sig{p}$. Then, for
  $q=p+|\psi|$, there is a simple formula $\psi'(\bar x)$ over a recoloring
  of $\sig{q}$ such that:

  for all graphs $\cG \in \Cp$ there is a graph $\cG' \in \Cq$ computable in time linear in
$\size{\cG}$ such that $\psi(\cG) = \psi'(\cG')$.

\end{lem}
\begin{proof}
  This is a simple consequence of transitivity. Any composition of two functions in
  $\cG$ represents a transitive pair of edges and becomes an single edge in
  the $1$-augmentation $\cH$ of $\cG$. Then $y=f(g(x))$ over $\cG$ is
  equivalent to $\bigvee_h y=h(x) \wedge P_{f,g,h}(x)$ over $\cH$, where $h$ is
  one of the new function introduced by the augmentation and
  the  newly introduced color $P_{f,g,h}$ holds for those nodes $v$, for which the
  $f(g(v))=h(v)$. As the nesting of compositions of functions is at most
  $|\psi|$, the result follows. The linear time computability is immediate
  from Lemma~\ref{lemma-frat-aug}.
\end{proof}

%

We make one more observation before proving the normal form:

\begin{lem}\label{lemma-order}
  Let $\cG \in \Cp$. Let $u$ be a node of $\cG$. Let $S$ be all the predecessors of $u$ in
  $\cG$ and set $q=p+\indeg{p}$. Let $\cG' \in \Cq$ be the $(q-p)$-th augmentation
  of $\cG$. There exists a linear order $<$ on $S$ computable from
  $\cG'$, such that for all $v,v' \in S$, $v<v'$ implies $v'=f(v)$ is an edge
  of $\cG'$ for some function $f$ from $\sigq$.
\end{lem}
\begin{proof}
  This is because all nodes of $S$ are fraternal and the size of $S$ is at most
  $\indeg{p}$. Hence, after one step of augmentation, all nodes of $S$ are
  pairwise connected and, after at most $\indeg{p}-1$ further augmentation steps, if
  there is a directed path from one node $u$ of $S$ to another node $v$ of $S$,
  then there is also a directed edge from $u$ to $v$. By induction on $|S|$ we show
  that there exists a node $u\in S$ such that for all $v \in S$ there is an
  edge from $v$ to $u$. If $|S|=1$ there is nothing to prove. Otherwise fix
  $v\in S$ and let $S'=S\setminus\set{v}$. By induction we get a $u$ in
  $S'$ satisfying the properties. If there is an edge from $v$ to $u$,
  $u$ also works for $S$ and we are
  done. Otherwise there must be an edge from $u$ to $v$. But then there is a
  path of length 2 from any node of $S'$ to $v$. By transitivity this means
  that there is an edge from any node of $S'$ to $v$ and $v$ is the node we are
  looking for.

  We then set $u$ as the minimal element of our order on $S$ and we repeat this
  argument with $S\setminus\set{u}$.
\end{proof}

Lemma~\ref{lemma-order} justifies the following definition. Let $p$ be a number
and let $q=p+\indeg{p}$. A $p$-type $\tau(x)$ is a quantifier-free formula over
the signature $\sig{q}$ with one free variable $x$ consisting of the
conjunction of a maximal consistent set of clauses of the form $f(g(x))=h(x)$
or $f(x) \neq x$. Given a node $u$ of some graph $\cG$ of $\C_p$, its $p$-type is
the set of clauses satisfied by $u$ in the $(q-p)$-th augmentation $\cG'$ of
$\cG$. From Lemma~\ref{lemma-order} it follows that the $p$-type of $u$ induces a linear order on its
predecessors in $\cG$. Indeed the predecessors of $u$ in $\cG$ can be deduced from the
$p$-type by looking at the clauses $f(x)\neq x$ where $f$ is a function symbol
from $\sig{p}$ and the linear order can be deduced from the clauses
$h(f(x))=g(x)$. Lemma~\ref{lemma-order} guarantees that these latter clauses
induce a linear order. In the sequel we denote this property as ``the $p$-type $\tau$ induces the order
$f_1(x) < f_2(x) < \cdots$'' and for $i<j$ we refer to the $h$ linking $f_i(x)$
to $f_j(x)$ as $h_{i,j}$.

 Note that for a given
$p$ there are only finitely many possible $p$-types and that each of them can
be specified with a conjunctive formula over $\sigq$.

We now state the normal form result.

\begin{prop}\label{prop-normalization}
Let $\phi(\bar x y)$ be a simple quantifier-free query over a recoloring of $\sigp$.
There exists $q$ that depends only on $p$ and $\phi$ and a quantifier-free query $\psi$
over a recoloring of $\sigq$ that is a disjunction of formulas:
\begin{equation}\label{eq-nf}
\psi_1(\bar x) \land \tau(y) \land \Delta^{=}(\bar x y) \land
\Delta^{\neq}(\bar x y),
\end{equation}
where  $\tau(y)$ implies a $p$-type of $y$; $\Delta^{=}(\bar x y)$ is either empty
or contains one clause of the form $y=f(x_i)$ or one clause of the form
$f(y)=g(x_i)$ for some $i$, $f$ and $g$; and $\Delta^{\neq}(\bar x y)$
contains arbitrarily many clauses of the form $y\neq f(x_i)$ or $f(y) \neq
g(x_j)$. Moreover, $\psi$ is such that:

for all $\cG \in \Cp$ there is a $\cG' \in \Cq$ computable in time linear in
$\size{\cG}$ with  $\phi(\cG) = \psi(\cG')$.
\end{prop}

\begin{proof}

Set $q$ as given by Lemma~\ref{lemma-order}.
We first put $\phi$ into a disjunctive normal form (\DNF) and in front of each such
disjunct we add a big disjunction over all possible $p$-types of $y$ (recall
that a type can be specified as a conjunctive formula). We deal with each
disjunct separately.

Note that each disjunct is a query over $\sig{q}$ of the form:
\begin{equation*}
\psi_1(\bar x) \land \tau(y) \land \Delta^{=}(\bar x y) \land
\Delta^{\neq}(\bar x y),
\end{equation*}
where all sub-formulas except for $\Delta^{=}$ are as desired.
Moreover, $\psi_1(\bar x)$, $\Delta^{=}(\bar x y)$ and
$\Delta^{\neq}(\bar x y)$ are in fact queries over $\sig{p}$.
At this point $\Delta^{=}$ contains arbitrarily many clauses of
the form $y=f(x_i)$ or $f(y)=g(x_i)$. If it contains at least
one clause of the form $y=f(x_i)$, we can replace each other
occurrence of $y$ by $f(x_i)$ and we are done.

Assume now that $\Delta^{=}$ contains several conjuncts of the form
$f_i(y)=g(x_k)$. Assume wlog that $\tau$ is such that $f_1(y) < f_2(y) <
\cdots$, where $f_1(y),f_2(y), \cdots$ are all the predecessors of $y$
from $\sig{p}$. Let $i_0$ be the smallest index $i$ such that a clause
of the form $f_i(y)=g(x_k)$ belongs to $\Delta^{=}$. We have
$f_{i_0}(y)=g(x_k)$ in $\Delta^{=}$ and recall that $\tau$
specifies for $i<j$ a function $h_{i,j}$ in $\sig{q}$ such that
$h_{i,j}(f_i(y))=f_j(y)$. Then, as $y$ is of type $\tau$, a clause of
the form $f_j(y)=h(x_{k'})$ with $i_0<j$ is equivalent to
$h_{i_0,j}(g(x_k))=h(x_{k'})$.
\end{proof}

\begin{myexa}\label{ex-a-nf}
Let us see what Lemma~\ref{lemma-simple} and
the normalization algorithm do for $p=0$ and some of the
disjuncts of the query of \eA~\ref{ex-a-enum-f}:

In the case of $f(g(x)) = y$ note that by transitivity, in the augmented
graph, this clause is equivalent to one of the form $y=h(x) \wedge P_{f,g,h}(x)$
(this case is handled by Lemma~\ref{lemma-simple}).

Consider now $\exists z~~ f(z) = x \wedge g(z) = y$. It will be convenient to
view this query when $z$ plays the role of $y$ in
Proposition~\ref{prop-normalization}. Notice that in this case it is not in
normal form as $\Delta^=$ contains two elements. However, the two edges $f(z)
= x$ and $g(z) = y$ are fraternal. Hence, after one augmentation step, a new edge is added
between $x$ and $y$ and we either have $y=h(x)$ or $x=h(y)$ for some $h$ in the
new signature.

Let $\tau_{h,f,g}(z)$ be $0$-type stating that $h(f(z))=g(z)$ and $\tau_{h,g,f}(z)$ be $0$-type stating that $h(g(z))=f(z)$.
It is now easy to see that the query $\exists z ~~f(z) = x \wedge g(z) = y$ is equivalent to
\begin{align*}
\exists z \bigvee_{h} &y=h(x) \land \tau_{h,f,g}(z) \land f(z)=x~~~ \vee\\ &x=h(y)
\land \tau_{h,g,f}(z) \land g(z)=y.
\end{align*}
\end{myexa}

\newcommand\witness{\textsc{Witness}\xspace}

\section{Model checking}\label{section-MC}

In this section we show that the model checking problem of \FO over a class of
structures with bounded expansion can be done in time linear in the size of the
structure. This gives a new proof of the result of~\cite{DKT11}.
Recall that by Lemma~\ref{lemma-structures-to-graphs} it is enough to consider
oriented graphs viewed as functional structures.

\begin{thmC}[\cite{DKT11}]\label{theo-model-check}
  Let \C be a class of graphs with bounded expansion and let $\psi$ be a
  sentence of \FO\@. Then, for all $\cG\in\C$, testing whether
  $\cG \models \psi$ can be done in time $O(\size{\cG})$.
\end{thmC}

The proof of Theorem~\ref{theo-model-check} is done using a quantifier
elimination procedure: given a query $\psi(\bar x)$ with at least one free
variable we can compute a quantifier-free query $\phi(\bar x)$ that is
``equivalent'' to $\psi$. Again, the equivalence should be understood modulo
some augmentation steps for a number of augmentation steps depending only on \C
and $|\psi|$. When starting with a sentence $\psi$ we end-up with $\phi$ being
a boolean combination of formulas with one variable. Those can be easily tested
in linear time in the size of the augmented structure, which in turn can be
computed in time linear from the initial structure by
Lemma~\ref{lemma-frat-aug}. The result follows. We now state precisely the quantifier elimination
step:

\begin{prop}\label{prop-quant-elim}
  Let \C be a class of graphs with bounded expansion witnessed by the
  function $\indegname$. Let $\psi(\bar x y)$ be a quantifier-free
  formula over a recoloring of $\sigp$. Then one can compute a $q$ and
  a quantifier-free formula $\phi(\bar x)$ over a recoloring of
  $\sigq$ such that:

  for all $\cG\in \Cp$ there is an $\cG' \in \Cq$ such that:
\begin{equation*}
\phi(\cG') = (\exists y \psi)(\cG)
\end{equation*}
Moreover, $\cG'$ is computable in time $O(\size{\cG})$.
\end{prop}

\begin{proof}
  In view of Lemma~\ref{lemma-simple} we can assume that $\psi$ is simple.
  We then apply Proposition~\ref{prop-normalization} to $\psi$ and
  $p$ and obtain a $q$ and an equivalent formula in DNF, where each disjunct has the
  special form given by~\eqref{eq-nf}. As disjunction and existential
  quantification commute, it is enough to treat each part of the disjunction
  separately.

  We thus assume that $\psi(\bar x y)$ is a quantifier-free conjunctive
  formula over a recoloring of $\sigq$ of the form~\eqref{eq-nf}:
\begin{equation*}
\psi_{1}(\bar x) \land \tau(y) \land \Delta^{=}(\bar x y) \land
\Delta^{\neq}(\bar x y).
\end{equation*}


Let's assume that the $p$-type $\tau$ satisfied by $y$ enforces $f_1(y) < f_2(y) <
\cdots$, where $f_1(y),f_2(y), \cdots$ are all the images of $y$ by a function
from $\sigp$ such that $f_i(y) \neq y$. Moreover, for each $i<j$, $\tau$ contains an atom of the form
$h_{i,j}(f_i(y))=f_j(y)$ for some function $h_{i,j}\in\sig{q}$.

We do a case analysis depending on the value of $\Delta^=$.
\begin{itemize}
    \item If $\Delta^=$ is $y=g(x_k)$ for some function $g$ and some $k$, then we replace
$y$ with $g(x_k)$ everywhere in $\psi(\bar x y)$ resulting in a formula
$\phi(\bar x)$ having obviously the desired properties.
    \item Assume now that $\Delta^=$ is of the form $f(y)=g(x_k)$. Without loss
of generality we can assume that $f$ is $f_{i_0}$ and $k=1$. In other words
$\Delta^=$ contains only the constraint $f_{i_0}(y)=g(x_1)$.

The general idea is to limit the quantification on $y$ to a finite set (whose size depends only on $\C$ and $\psi$),
depending only on $x_1$. We then encode these sets using suitable extra colors.
To do this, for each node $w$ we first compute a set $\witness(w)$ such that for each
tuple $\bar v$ we have $\cG_q\models \exists y~~ \psi(\bar v y)$ iff $\cG_q \models
\exists y \in \witness(g(v_1))~~ \psi(\bar v y)$. Moreover, for all $w$, $|\witness(w)|\leq N$ where $N$ is a number
depending only on $p$. We then encode these witness sets using suitable extra colors.

The intuition behind the $\witness$ set is as follows. Assume first that
$\Delta^{\neq}$ is empty. Then we only need to test the existence of $y$ such
that $f_{i_0}(y)=g(x_1)$. To do so, we scan through all nodes $u$, test if
$\tau(u)$ holds and if so we add $u$ to $\witness(f_{i_0}(u))$ if this set is
empty and do nothing otherwise. Clearly each witness set has size at most one
and the desired properties are verified. Moreover if $v$ is in
$\witness(g(x_1))$ then $f_{i_0}(v)=g(x_1)$. Therefore it is then enough to
color with a new color red all nodes having a non-empty witness set and
$\exists y ~ \tau(y) \land f_{i_0}(y)=g(x_1)$ is then equivalent to
$\text{red}(g(x_1))$.

The situation is slightly more complicated if $\Delta^{\neq}$ is not
empty. Assume for instance that $\Delta^{\neq}$ contains only constraints of
the form $y\neq h(x_k)$. Then the previous procedure does not work because
$\witness(g(x_1))$ may be such that it is equal to $h(x_k)$. However
there are only $c$ nodes that we need to avoid, where $c$ depends only on the
formula, hence if $\witness(g(x_1))$ contains at least $c+1$ nodes we are sure
that at least one of them will satisfy all the inequality constraints. We
implement this by scanning through all nodes $u$, test if $\tau(u)$ holds and
if so we add $u$ to $\witness(f_{i_0}(u))$ if this set has a size smaller or
equal to $c$ do nothing otherwise. The difficulty is to encode this set into
the formula. If the witness set is of size $c+1$ one of its element must make
all inequalities true hence a new color as before does the job. When the set
has a smaller size we need to test each of its elements against the
inequalities.  For this we introduce a predicates $Q_i$, and add a node $u$ to
$Q_i$ if $u$ has been added as the $i^{th}$ element in $\witness(f_{i_0}(u))$.
As before any element $y$ in $\witness(g(x_1))$ is such that
$f_{i_0}(y)=g(x_1)$. It remains to test whether the $i^{th}$ such element
satisfies $y \neq h(x)$. In other words whether $h(x)$ is the $i^{th}$ witness
of $g(x_1)$ or not. It is easy to check that the $i^{th}$ witness of $g(x_1)$
is $h(x)$ iff $Q_i(h(x_k)) \land
f_{i_0}(h(x_k))=g(x_1)$.

The general case, when $\Delta^{\neq}$ contains also clauses of the form
$h_1(y) \neq h_2(x_k)$ is more complex and require an even bigger witness set
but this is essentially what we do.

\paragraph{Computation of the Witness function}

We start by initializing $\witness(v)=\emptyset$ for all $v$.

We then successively investigate all nodes $u$ of $\cG_q$ and do the
following.  If $\cG_q \models \lnot\tau(u)$ then we move on to the
next $u$. If $\cG_q \models \tau(u)$ then let $u_1,\dots, u_l$ be the
current value of $\witness(f_{i_0}(u))$ --- if $\witness(f_{i_0}(u))$ is empty then
we add $u$ to this set and move on to the next node of $\cG_q$.

Let $\beta_p$ be $\sigfp (\sigfp+1)|\bar x|+1$.

Let $i$ be minimal such that there exists $j$ with $f_i(u_j) = f_i(u)$ (notice
that $i\leq i_0$). Note that because
$f_j(w)=h_{i,j}(f_i(w))$ for all $w$ verifying $\tau$ and all $j>i$, this
implies that $u$ and $u_j$ agree on each $f_j$ with $j \geq i$ and disagree on
each $f_j$ with $j<i$.

Let $S_i=\set{f_{i-1}(u_j)
  ~|~ f_i(u_j)=f_i(u)}$, where $f_0(u_j)$ is $u_j$ in the case where $i=1$. If
$|S_i| < \beta_p$ then we add $u$ to $\witness(f_{i_0}(u))$.

\paragraph{Analysis of the Witness function}
Clearly the algorithm computing the witness function runs in linear time.

Moreover, for each node $v$, $\witness(v)$ can be represented as the leaves of a tree of
depth at most $\sigfp$ and of width $\beta_p$. To see this, notice that all nodes $u$ of $\witness(v)$ are such that
$f_{i_0}(u)=v$. Note also that if two nodes $u$ and $u'$ satisfying $\tau$
share a predecessor, $f_i(u)=f_i(u')$, then for all $j>i$, $u$ and $u'$ agree
on $f_j$ as $f_j=h_{i,j} \circ f_i$ for all nodes
satisfying $\tau$. The depth of the least common ancestor of two nodes $u$ and $u'$ of
$\witness(v)$ is defined as the least $i$ such that $u$ and $u'$ agree on
$f_i$. One can then verify that by construction of $\witness(v)$ the tree has
the claimed sizes. Hence the size of $\witness(v)$ is
bounded by $\beta_p^{\sigfp+1}$.

We now show that for each tuple $\bar v$ and each node $u$ such that $\cG_q
\models \psi(\bar v u)$ there is a node $u'$ in $\witness(g(v_1))$ such that
$\cG_q \models \psi(\bar v u')$.

To see this assume $\cG_q \models \psi(\bar v u)$. If $u\in \witness(g(v_1))$ we are
done. Otherwise note that $f_{i_0}(u)=g(v_1)$ and that
$\cG_q \models \tau(u)$.
Let $i$ and $S_i$ be as described in the algorithm when
investigating $u$. As $u$ was not added to $\witness(f_{i_0}(u))$, we must have
$|S_i| > \beta_p$. Let $u_{1},\dots,u_{\beta_p}$ be the
elements of $\witness(g(v_1))$ providing $\beta_p$ pairwise different values for
$f_{i-1}$.  Among these, at most $\sigfp |\bar v|$ of them may be of the form
$f_j(v_l)$ for some $j$ and $l$ as each $v_l$ has at most $\sigfp$ predecessors.
Notice that for all $j>i$ and all $\ell$, $u$ agrees with $u_\ell$ on $f_i$ and therefore
they also agree on $f_j$ for $j>i$ as $f_j=h_{i,j} \circ f_i$ for all nodes
satisfying $\tau$. When $j<i$ the values of $f_j(u_\ell)$ and $f_j(u_{\ell'})$ must be
different if $\ell\neq \ell'$ as otherwise $u_\ell$ and $u_{\ell'}$ would also agree on
$f_{i-1}$ as $f_{i-1}=h_{j,{i-1}} \circ h_{i-1}$ for all nodes satisfying $\tau$.
Therefore, for each $\ell$ and each $j<i$ there are at most $\sigfp$
of the $u_\ell$ such that $f_j(u_\ell)$ is a predecessor of $v_l$.

Altogether, at most $\sigfp^2|\bar v| + \sigfp|\bar v|$ nodes $u_\ell$ may falsify an
inequality constraint. As $\beta_p$ is strictly bigger than that, one of the
$u_\ell$ is the desired witness.

\paragraph{Recoloring of $\cG_q$}

Based on \witness we recolor $\cG_q$ as follows. Let $\gamma_p={(\beta_p +
1)}^{\sigfp+1}$.  For each $v\in \cG_q$, the $i^{th}$ witness of $v$ is the
$i^{th}$ element inserted in  $\witness(v)$ by the algorithm.

For each $i \leq \gamma_p$ we introduce a new unary predicate $P_i$ and for each
$u\in \cG_q$ we set $P_i(u)$ if $\witness(u)$ contains at least $i$
elements.

For each $i\leq \gamma_p$, we introduce a new unary predicate $Q_{i}$ and for
each $v\in \cG_q$ we set $Q_{i}(v)$ if the $i^{th}$ witness of $f_{i_0}(v)$ is
$v$.

For each $i\leq \gamma_p$ and each $h,h'\in \sigfq$ we introduce a new unary
predicate $P_{i,h,h'}$ and for each $v\in \cG_q$ we set $P_{i,h,h'}(v)$ if the
$i^{th}$ witness of $h(v)$ is an element $u$ with $h'(u)=v$.

We denote by $\cG'$ the resulting graph and notice that it can be
computed in linear time from $\cG$.


\paragraph{Computation of $\phi$}

We now replace $\psi(\bar x,y)$ by the following formula:
\begin{equation*}
\bigvee_{i\leq \gamma_p} \psi_{1}(\bar x) \land \psi^i(\bar x)
\end{equation*}
where $\psi^i(\bar x)$ checks that the $i^{th}$ witness of $g(x_1)$ makes the
initial formula true.

Notice that if $y$ is the $i^{th}$ witness of $g(x_1)$ then
$f_{i_0}(y)=g(x_1)$. Hence the equality
$f_j(y) = h(x_k)$ with $j<i_0$ is equivalent over $\cG'$ to
$h_{j,i_0}(h(x_k))=g(x_1) \land P_{i,h_{j,i_0},f_j}(h(x_k))$
 and the equality
$y = h(x_k)$ is equivalent over $\cG'$ to $f_{i_0}(h(x_k))=g(x_1) \land
Q_{i}(h(x_k))$. From the definition of $p$-type, the equality
$f_j(y) = h(x_k)$ with $j>i_0$ is equivalent to $h_{i_0,j}(g(x_1)) = h(x_k)$.

This implies that $\psi^i(\bar x)$ can be defined as
{
\begin{eqnarray*}
 P_i(g(x_1)) &\land& \!\!\!\!\!\!\!\!\!\!\!\!\!\!\!\bigwedge_{\substack{f_j(y) \neq h(x_k) \in \Delta^{\neq}\\ j<i_0}}  \!\!\!\!\!\!\!\!\!\!\!\!\!\lnot
 \big(h_{j,i_0}(h(x_k))=g(x_1) \land
 P_{i,h_{j,i_0},f_j}(h(x_k))\big)\\
 &\land& \bigwedge_{\substack{f_j(y) \neq h(x_k) \in \Delta^{\neq}\\j \geq i_0}}
 h_{i_0,j}(g(x_1)) \neq h(x_k)\\
 &\land& \bigwedge_{y \neq h(x_k) \in \Delta^{\neq}}  \!\!\!\!\!\!\!\!\!\lnot
\big(f_{i_0}(h(x_k))=g(x_1) \land Q_{i}(h(x_k))\big).
\end{eqnarray*}
}


    \item It remains to consider the case when $\Delta^=$ is empty. This is a simpler version of the
previous case, only this time it is enough to construct a set $\witness$
which does not depend on $v$.
It is constructed as in the previous case and verifies:
for all tuples $\bar v$ of $\cG_q$, if $\cG_q\models \psi(\bar v u)$
for some node $u$, then there is a node $u'\in \witness$ such that
$\cG_q \models \psi(\bar v u')$. Moreover, $|\witness|\leq \gamma_p$. We then
argue as in the previous case.
\qedhere
\end{itemize}
\end{proof}

\begin{myexa}\label{ex-a-qelim}
Consider one of the quantified formulas as derived by \eA~\ref{ex-a-nf}:
\begin{equation*}
\exists z~~y=h(x) \land \tau_{h,f,g}(z) \land f(z) = x.
\end{equation*}
The resulting quantifier-free query has the form:
\begin{equation*}
P(x) \wedge h(x) = y
\end{equation*}
where $P(x)$ is a newly introduced color saying ``$\exists z~~ \tau_{h,f,g}(z) \land f(z) = x$''.
The key point is that this new predicate can be computed in linear time by iterating through all
nodes $z$, testing whether $\tau_{h,f,g}(z)$ is true and, if this is the case,
coloring $f(z)$ with color $P$.
\end{myexa}

%

Applying the quantifier elimination process from inside out using
Proposition~\ref{prop-quant-elim} for each step and then applying
Lemma~\ref{lemma-simple} to the result yields:

\begin{thm}\label{thm-quantifier-elim}
Let \C be a class of graphs with bounded expansion. Let $\psi(\bar{x})$ be a query of \FO over a recoloring of $\sig{0}$ with at least
one free variable. Then one can compute a $p$ and a simple
quantifier-free formula $\phi(\bar x)$ over a recoloring of $\sigp$ such that:

  $\forall \cG\in\C$, we can construct in time $O(\size{\cG})$
  a graph $\cG' \in \Cp$ such that
\begin{equation*}
 \phi(\cG') = \psi(\cG).
\end{equation*}
\end{thm}

We will make use of the following useful consequence of
Theorem~\ref{thm-quantifier-elim}:

\begin{cor}\label{testing-FO}
  Let \C be a class of graphs with bounded expansion and let $\psi(\bar{x})$ be a formula of \FO over $\sig{0}$ with at least one free variable.
  Then, for all $\cG\in\C$, after a preprocessing in time $O(\size{\cG})$,
  we can test, given $\bar u$ as input, whether $\cG \models \psi(\bar u)$
  in constant time.
\end{cor}
\begin{proof}
  By Theorem~\ref{thm-quantifier-elim} it is enough to consider quantifier-free
  simple queries. Hence it is enough to consider a query consisting in a single
  atom of either $P(x)$ or $P(f(x))$ or $x=f(y)$ or $f(x)=g(y)$.

  During the preprocessing phase we associate to each node $v$ of the input
  graph a list $L(v)$ containing all the predicates satisfied by $v$  and all
  the images of $v$ by a function symbol from the signature. This can be
  computed in linear time by enumerating all relations of the database and
  updating the appropriate lists with the corresponding predicate or the
  corresponding image.

  Now, because we use the RAM model, given $u$ we can in constant time recover
  the list $L(u)$. Using those lists it is immediate to check all atoms of the
  formula in constant time.
\end{proof}

Theorem~\ref{theo-model-check} is a direct consequence of
Theorem~\ref{thm-quantifier-elim} and Corollary~\ref{testing-FO}: Starting with
a sentence, and applying Theorem~\ref{thm-quantifier-elim} for eliminating
quantifiers from inside out we end up with a Boolean combination of formulas
with one variable. Each such formula can be tested in $O(\size{\cG})$ by
iterating through all nodes $v$ of $\cG$ and in constant time (using
Corollary~\ref{testing-FO}) checking if $v$ can be substituted for the
sole existentially quantified variable.

On top of Theorem~\ref{theo-model-check} the following corollary is immediate
from Theorem~\ref{thm-quantifier-elim} and Corollary~\ref{testing-FO}:

\begin{cor}\label{unary-FO}
  Let \C be a class of graphs with bounded expansion and let $\psi(x)$ be a
  formula of \FO over $\sig{0}$ with one free variable. Then, for all
  $\cG\in\C$, computing the set $\psi(\cG)$ can be done in time $O(\size{\cG})$.
\end{cor}

\newcommand{\n}[1]{\ensuremath{\text{next}(#1)}\xspace}
\newcommand{\nemp}[1]{\ensuremath{\textsc{next}_{\emptyset}(#1)}\xspace}
\newcommand{\ns}[2]{\ensuremath{\textsc{next}_{#1}(#2)}\xspace}

\newcommand{\nspold}[1]{\ensuremath{\ns{f_{1},S_{1}, \ldots, f_{\sigfq}, S_{\sigfq}}{#1}}}
\newcommand{\nsppold}[1]{\ensuremath{\ns{f_{1},S'_{1}, \ldots, f_{\sigfq}, S'_{\sigfq}}{#1}}}
\newcommand{\nspppold}[1]{\ensuremath{\ns{f_{1},S''_{1}, \ldots, f_{\sigfq}, S''_{\sigfq}}{#1}}}
\newcommand{\sthold}[1]{\ensuremath{\ns{f_{1},S_{1}, \ldots, f_{i}, S_{i} \cup \set{u_{i}}, \ldots, f_{\sigfq}, S_{\sigfq}}{#1}}}

\newcommand{\nspnew}[1]{\ns{\vec{S}}{#1}}
\newcommand{\nsppnew}[1]{\ns{\vec{S'}}{#1}}
\newcommand{\nspppnew}[1]{\ns{\vec{S''}}{#1}}
\newcommand{\sthnew}[1]{\ns{\vec{S}[S_{i} += \set{u_i}]}{#1}}
\newcommand{\spthnew}[1]{\ns{\vec{S'}[S'_{i} += \set{u_i}]}{#1}}

\newcommand{\nsp}[1]{\nspnew{#1}}
\newcommand{\nspp}[1]{\nsppnew{#1}}
\newcommand{\nsppp}[1]{\nspppnew{#1}}
\newcommand{\sth}[1]{\sthnew{#1}}
\newcommand{\spth}[1]{\spthnew{#1}}

\newcommand{\bq}{\ensuremath{\beta_{q}}\xspace}
\newcommand{\SC}{\ensuremath{\textsc{SC}_{L}}\xspace}
\newcommand{\SCu}[1]{\ensuremath{\textsc{SC}_{L}(#1)}\xspace}
\newcommand{\sm}{\preceq}

\section{Enumeration}\label{section-enum}

In this section we consider first-order formulas with free variables and show
that we can enumerate their answers with constant delay over any class with
bounded expansion. Moreover, assuming a linear order on the domain of the input
structure, we will see that the answers can be output in the lexicographical
order. As before we only state the result for graphs, but it immediately
extends to arbitrary structures by Lemma~\ref{lemma-structures-to-graphs}.

\begin{thm}\label{theo-enum}
  Let \C be a class of graphs with bounded expansion and let $\phi(\bar x)$ be a
  first-order query. Then the enumeration problem of $\phi$ over
  \C is in \CDlin.
  \\ Moreover, in the presence of a linear order on the vertices of the input
  graph, the answers to $\phi$ can be output in lexicographical order.
\end{thm}

The proof of Theorem~\ref{theo-enum} is by induction on the number of free variables of $\phi$.
The unary case is done by Corollary~\ref{unary-FO}. The inductive case is a
simple consequence of the following:

\begin{prop}\label{prop-enum}
  Let \C be a class of graphs with bounded expansion and let $\phi(\bar x y)$ be a
  first-order query or arity 2 or more. Let $\graph$ be a graph of $\C$. Let $<$ be any linear order on the nodes of $\graph$.
  After a preprocessing working in time linear in the size of $\graph$
  we can, on input a tuple $\bar a$ of nodes of $\graph$,
  enumerate with constant delay and in the order given by $<$ all $b$ such that $\graph
  \models \phi(\bar a b)$ or answer \textsc{nill} if no such $b$ exists.
\end{prop}

\begin{proof}
  Fix a class \C of graphs with bounded expansion and a query $\phi(\bar x y)$
  with $k\geq 2$ free variables. Let $\cG$ be the functional representation of the input graph and $V$ be its set
  of vertices. Let $<$ be any order on $V$.

During the preprocessing phase, we apply Theorem~\ref{thm-quantifier-elim} to
get a simple quantifier-free query $\varphi(\bar x y)$ and a structure
$\cG'\in\C_p$, for some $p$ that does not depend on $\cG$, such that
$\varphi(\cG')=\phi(\cG)$ and $\cG'$ can be computed in linear time from $\cG$.

Furthermore we normalize the resulting simple quantifier-free query using
Proposition~\ref{prop-normalization}, and obtain an equivalent quantifier-free
formula $\psi$ and a structure $\cG''\in\C_q$, where $q$ depends only on
$p$ and $\varphi$, $\cG''$ can be computed in linear time from $\cG'$,
$\varphi(\cG')=\psi(\cG'')$ and $\psi$ is a disjunction of formulas
of the form~\eqref{eq-nf}:
\begin{equation*}
\psi_1(\bar x) \land \tau(y) \land \Delta^{=}(\bar x y) \land
\Delta^{\neq}(\bar x y),
\end{equation*}
where  $\Delta^{=}(\bar x y)$ is either empty or
contains one clause of the form $y=f(x_i)$ or one clause of the form
$f(y)=g(x_i)$ for some $i$, $f$ and $g$; and $\Delta^{\neq}(\bar x y)$
contains arbitrarily many clauses of the form $y\neq f(x_i)$ or $f(y) \neq g(x_j)$.

In view of Lemma~\ref{disjunct-enum} it is enough to treat each
disjunct separately. In the sequel we then assume that $\psi$ has the form
described in~\eqref{eq-nf}. We let $\psi'(y)$ be the formula $\exists \bar{x} \psi(\bar x y)$ and $\psi''(\bar x)$ the formula $\exists y \psi(\bar x y)$.

If $\Delta^{=}$ contains an equality of the form $y=f(x_i)$ we then use
Corollary~\ref{testing-FO} and test whether $\cG''\models \psi(\bar a f(a_i))$ and
we are done as $f(a_i)$ is the only possible solution for $\bar a$.

Assume now that $\Delta^{=}$ is either empty or of the form $f(y)=g(x_i)$.

We first precompute the set of possible candidates for $y$ (i.e., those $y$
satisfying $\psi'$) and distribute this set within their images by $f$. In other
words we define a function $L: V \rightarrow 2^V$ such that
\[
    L(w) =\set{u ~|~ w=f(u) \land u \in \psi'(\cG'')}.
\]
In the specific case where $\Delta^=$ is
empty we pick an arbitrary  node $w_0$ in $\cG''$ and set
$L(w_0)=\psi'(\cG'')$ and $L(w)=\emptyset$ for $w\neq w_0$. This can be done in linear
time by the following procedure. We first use Corollary~\ref{unary-FO} and
compute in linear time the set $\psi'(\cG'')$. We next initialize $L(w)$ to $\emptyset$ for each
$w\in V$. Then, for each $u\in\psi'(\cG'')$, we add $u$ to the set $L(f(u))$.

Let $W$ be the function from $V^{k-1}$ to $V$ such that $W(\bar v) =
g(v_i)$. In the specific case where $\Delta^=$ is
empty we set $W(\bar v)=w_0$, where $w_0$ is the node chosen above.

Notice that for each $\bar v u$, $\cG'' \models \psi(\bar v u)$ implies $u \in L(W(\bar v))$ and  if $u \in
L(W(\bar v))$ then $\Delta^{=}(\bar v u)$ is true. Hence, given $\bar a$ it
remains to enumerate within $L(W(\bar a))$ the nodes $b$ satisfying
$\Delta^{\neq}(\bar a b)$.

To do this with constant delay, it will be important to jump from an element $u$ of $L(w)$ to the smallest
(according to $<$) element $u' \geq u$ of $L(w)$ satisfying the inequality constraints.

For this we define for $S_{1}, \ldots, S_{\sigfq} \subseteq V$ the element $\nspold{u}$ to be the
first element $w \geq u$ of $L(f(u))$\footnote{In order to simplify the
  notations we consider explicitly the case where $\Delta^=$ is not empty. If
  empty then $L(f(u))$ should be replaced by $L(w_0)$.} such that $f_{1}(w) \notin S_{1}, \ldots, \text{ and }
f_{\sigfq}(w) \notin S_{\sigfq}$. If such $w$ does not exist, the value
of $\nspold{u}$ is \nil. When all $S_{i}$ are empty, we write $\nemp{u}$ and by the above
definitions we always have $\nemp{u} = u$.
We denote such functions as \emph{shortcut pointers of $u$}.
The \emph{size} of a shortcut pointer $\nspold{u}$ is the sum of sizes of the sets $S_{i}$.

In order to avoid writing too long expressions containing shortcut pointers, we
introduce the following abbreviations:
\begin{itemize} \itemsep1pt \parskip0pt \parsep0pt
	\item $\nspold{u}$ is denoted with $\nspnew{u}$,
	\item $\sthold{u}$ is denoted with $\sthnew{u}$.
\end{itemize}

\medskip

\noindent
Set $\gamma_q = (k-1) \cdot \sigfq^2$.

\medskip

Computing all shortcut pointers of size $\gamma_q$ would take more than linear time.
We therefore only compute a subset of those, denoted \SC, that will be sufficient
for our needs. \SC is defined in an inductive manner. For all $u$
such that $u \in L(f(u))$, $\nemp{u} \in \SC$.
Moreover, if the shortcut pointer $\nsp{u}\in\SC$ is not $\nil$ and has a size smaller
than $\gamma_q$, then, for each $i$, $\sth{u}\in\SC$, where $u_i=f_i(\nsp{u})$.
We then say that $\nsp{u}$ is the \emph{origin} of $\sth{u}$.
Note that \SC contains all the shortcut pointers of the form
$\ns{f_{i}, \set{f_{i}(u)}}{u}$ for $u\in L(f(u))$ and these are exactly the
shortcut pointers of $u$ of size $1$. By $\SCu{u} \subseteq \SC$ we denote
the shortcut pointers of $u$ that are in \SC\@.

The set \SC contains only a constant number of shortcut pointers for each node $u$.

\begin{clm}\label{claim-SC-size}
There exists a constant $\zeta(q,k)$ such that for every
node $u$ we have $|\SCu{u}| \leq \zeta(q,k)$.
\end{clm}

\begin{proof}
The proof is a direct consequence of the recursive definition of $\SCu{u}$.
Fix $u$. Note that there is exactly $1$ shortcut pointer of $u$
of size $0$ (namely $\nemp{u}$) and $\sigfq$ shortcut pointers of $u$
of size $1$. By the definition of $\SC$, any shortcut pointer
$\nsp{u}$ can be an origin of up to $\sigfq$ shortcut pointers
of the form $\sth{u}$, where $u_{i} = f_{i}(\nsp{u})$ and the
size of $\sth{u}$ is the size of $\nsp{u}$ plus $1$. This way we see that $\SCu{u}$
contains up to $\sigfq^{2}$ shortcut pointers of size $2$ and, in general,
up to $\sigfq^{s}$ shortcut pointers of size $s$. As the maximal size of
a computed shortcut pointer is bounded by $\gamma_q$, we have
$|\SCu{u}| \leq \sum_{0 \leq i \leq \gamma_q} \sigfq^{i}$. Both $\sigfq$ and
$\gamma_q$ depend only on $q$ and $k$, which concludes the proof.
\end{proof}

Moreover $\SC$ contains all what we need to know.

\begin{clm}\label{claim-SC-okjump}
Let $\nsp{u}$ be a shortcut pointer of size not greater than $\gamma_q$.
Then there exists $\nspp{u} \in \SC$ such that $\nsp{u} = \nspp{u}$.
Moreover, such $\nspp{u}$ can be found in constant time.
\end{clm}

\begin{proof}
If $\nsp{u} \in \SC$, then we have nothing to prove. Assume then
that $\nsp{u} \notin \SC$.
We write $\nsppold{u} \sm \nspold{u}$ if for each
$1 \leq i \leq \sigfq$ we have $S'_{i} \subseteq S_{i}$. Note that for a given
$u$ the $\sm$ relation is a partial order on the set of shortcut pointers of $u$.
A trivial observation is that if $\nsppold{u} \sm \nspold{u}$, then
$\nsppold{u} \leq \nspold{u}$.

 Let $\nspp{u} \in \SC$ be a maximal
in terms of size shortcut pointer of $u$ such that $\nspp{u} \sm \nsp{u}$. Such
a shortcut pointer always exists as $\nemp{u} \sm \nsp{u}$ and
$\nemp{u} \in \SC$. Note that the size of $\nspp{u}$ is strictly
smaller than the size of $\nsp{u}$, so it is strictly smaller than $\gamma_q$.
One can find $\nspp{u}$ by exploring all the shortcut pointers of $u$ in
$\SCu{u}$. This can be done in constant time using Claim~\ref{claim-SC-size}.

We now claim that $\nsp{u} = \nspp{u}$.

Let $v = \nspp{u}$. We know that $v \leq \nsp{u}$.
Assume now that there would exists $1 \leq i \leq \sigfq$
such that $f_{i}(v) \in S_{i}$. Then we have that $\spth{u} \in \SC$, where
$u_i=f_i(v)$, and this contradicts the maximality of $\nspp{u}$. This means that such an $i$ does
not exist and concludes the fact that $\nsp{u} = \nspp{u}$.
\end{proof}


The following claim guarantees that \SC can be computed in linear time and has
therefore a linear size.

\begin{clm}\label{claim-SC-lin}
  \SC can be computed in time linear in $\size{\cG''}$.
\end{clm}
\begin{proof}
  For every $u$ we compute $\SC(u)$ in time linear in the size of $\SC(u)$.
  By Claim~\ref{claim-SC-size} the total time is linear in the
  size of $V$ as claimed.

  The computation of $\SC(u)$ is done in reverse order on $u$: Assuming
  $\SC(v)$ has been computed for all $v > u$ we compute $\SC(u)$ time linear in
  the size of $L(u)$.

  Note that we only care to compute $\nsp{u}$ when $u\in L(u)$.

  Consider a node $u$. If $u$ is the maximaum vertex then all $\nsp{u}$ are \nil.

  Assume now that $u$ is not the maximum vertex and that for all $v>u$ $\SC(v)$
  has been computed. If $u$ does not belong to $L(u)$ we do nothing. If
  $u\in L(u)$ we set $\nemp{u} = u$ and we construct $\SC(u)$ by induction on
  the size.

  Assume $\nsp{u} \in \SC$ has already been computed. Let $v=\nsp{u}$ and
  assume it is not $\nil$. Let
  $\vec{S'}=\vec{S}[S_i +\!\!= \set{f_i(v)}]$ and assume $\nspp{u} \in \SC(u)$.

  It is easy to see that $\nspp{u}=\nspp{v}$. By Claim~\ref{claim-SC-okjump} we
  can obtain in constant time $\nsppp{v} \in \SC(v)$ such that
  $\nspp{v}=\nsppp{v}$. As $v >u$ the value of $\nsppp{v}$ has already been computed.
\end{proof}

The computation of \SC concludes the preprocessing phase and it follows from
Claim~\ref{claim-SC-lin} that it can be done in linear time.
We now turn to the enumeration phase.

Assume we are given $\bar a$. In view of Corollary~\ref{testing-FO} we can
without loss of generality assume that $\bar a$ is such that
$\graph \models \psi''(\bar a)$. If not we simply return \textsc{nill} and stop
here.

By construction we know that all nodes $b$ such that
$\cG'' \models \psi(\bar a b)$ are in $L=L(W(\bar a))$. Recall also that all
elements $b\in L$ make $\tau(b) \wedge \Delta^{=}(\bar a b)$ true.  For
$1 \leq i \leq \sigfq$ we set
$S_{i} = \set{g(v_{j}): g(x_{j}) \neq f_{i}(y) \text{ is a conjunct of }
  \Delta^{\neq}}$.
Starting with $b$ the first node of the sorted list $L$, we apply the following
procedure:

\begin{enumerate} \itemsep1pt \parskip0pt \parsep0pt
\item\label{start} If $b$ is not $\nil$, let $\nspp{b}$ be the shortcut pointer from the
	application of Claim~\ref{claim-SC-okjump} to $\nsp{b}$. Set $b' = \nspp{b}$.
	If $b' = \nil$, stop here.

\item\label{output-ok} If $\cG'' \models \psi(\bar a b')$, output
	$b'$.

\item\label{output-bad} Reinitialize $b$ to the successor of $b'$ in $L$
        and continue with Step~\ref{start}.
\end{enumerate}

\noindent
We now show that the algorithm is correct.

The algorithm clearly outputs only solutions as it tests whether
$\cG'' \models \psi(\bar a b')$ before outputting~$b'$.

By the definition of sets $S_{i}$ and $\nsp{b}$, for each $b \leq v < b'$
there is a $i$ and $j$ such that $g(a_{j}) = f_{i}(v)$ and
$g(x_{j}) \neq f_{i}(y)$ is a conjunct of $\Delta^{\neq}$. This way the
algorithm does not skip any solutions at Step~\ref{start} and so
it outputs exactly all solutions.

It remains to show that there is a constant time between any two outputs.
Step~\ref{start} takes constant time due to Claim~\ref{claim-SC-okjump}. From
there the algorithm either immediately outputs a solution at
Step~\ref{output-ok} or jumps to Step~\ref{output-bad}. In the second case,
this means that $\cG'' \not\models \psi(\bar a b')$, but from the definitions
of list $L$, sets $S_{i}$ and shortcut pointers $\nsp{b}$ this can only happen
if $\Delta^{\neq}$ is falsified and this is because of an inequality of the
form $y \neq g(x_{j})$ for some suitable $g$ and $j$ (where $g$ may possibly be
identity). Hence $b'=g(a_j)$. As all the elements on $L$ are
distinct, the algorithm can skip over Step~\ref{output-ok} up to
$(k-1) \cdot (\sigfq+1)$ times for each tuple $\bar a$ (there are up to that
many different images of nodes from $\bar a$ under $\sigfq$ different
functions). The delay is therefore bounded by $k \cdot (\sigfq+1)$ consecutive applications of Claim~\ref{claim-SC-okjump}.

As the list $L$ was sorted with respect to the linear order on the domain,
it is clear that the enumeration procedure outputs the set of solutions in
lexicographical order.

This concludes the proof of the theorem.
\end{proof}

\section{Counting}

In this section we investigate the problem of counting the number of solutions
to a query, i.e., computing $|q(\cD)|$. As usual we only state and prove our
results over graphs but they generalize to arbitrary relational structures
via Lemma~\ref{lemma-structures-to-graphs}.

\begin{thm}\label{counting}
Let \C be class of graphs with bounded expansion and let $\phi(\bar x)$
be a first-order formula. Then, for all $\cG \in \C$, we can compute $|\phi(\cG)|$
in time $O(\size{\cG})$.
\end{thm}

\begin{proof}
The key idea is to prove a weighted version of the desired result. Assume
$\phi(\bar x)$ has exactly $k$ free variables and for $1 \leq i \leq k$ we have
functions $\#_{i} : V \rightarrow \Nat$. We will compute in time linear in
$\size{\cG}$ the following number:

\begin{equation*}
\cw{\phi}{\cG}{\#} := \sum_{\bar u \in
  \phi(\cG)} \prod_{1 \leq i \leq k} \#_{i}(u_{i}).
\end{equation*}

By setting all $\#_{i}$ to be constant functions with value $1$ we get the
regular counting problem. Hence Theorem~\ref{counting} is an immediate
consequence of the next lemma.

\begin{lem}\label{ct-lemma}
Let \C be class of graphs with bounded expansion and let $\phi(\bar x)$
be a first-order formula with exactly $k$ free variables.\\ For $1 \leq i \leq k$
let $\#_{i} : V \rightarrow \Nat$ be functions such that for each $v$
the value of $\#_i(v)$ can be computed in constant time.\\
Then, for all $\cG \in \C$, we can compute $\cw{\phi}{\cG}{\#}$
in time $O(\size{\cG})$.
\end{lem}

\begin{proof}
The proof is by induction on the number of free variables.

The case $k = 1$ is trivial: in time linear in $\size{\cG}$ we compute
$\phi(\cG)$ using Corollary~\ref{unary-FO}. By hypothesis, for each $v \in \phi(\cG)$, we
can compute the value of $\#_{1}(v)$ in constant time. Therefore the value
\begin{equation*}
\cw{\phi}{\cG}{\#} = \sum_{v \in \phi(\cG)} \#_{1}(v)
\end{equation*}
can be computed in linear time as desired.

Assume now that $k > 1$ and that $\bar x$ and $y$ are the free variables of
$\phi$, where $|\bar x|=k-1$.

We apply Theorem~\ref{thm-quantifier-elim} to get a simple quantifier-free query
$\varphi(\bar x y)$ and a structure $\cGp\in\Cp$, for some $p$ that does not
depend on $\cG$, such that $\varphi(\cGp)=\phi(\cG)$ and $\cGp$ can be
computed in linear time from $\cG$. Note that $\cw{\phi}{\cG}{\#}=\cw{\varphi}{\cGp}{\#}$,
so it is enough to compute the latter value.

We normalize the resulting simple quantifier-free query using
Proposition~\ref{prop-normalization}, and obtain an equivalent quantifier-free
formula $\psi$ and a structure $\cGpp\in\Cq$, where $q$ depends only on
$p$ and $\varphi$, $\cGpp$ can be computed in linear time from $\cGp$,
$\varphi(\cGp)=\psi(\cGpp)$ and $\psi$ is a disjunction of formulas
of the form~\eqref{eq-nf}:
\begin{equation*}
\psi_1(\bar x) \land \tau(y) \land \Delta^{=}(\bar x y) \land
\Delta^{\neq}(\bar x y),
\end{equation*}
where  $\Delta^{=}(\bar x y)$ is either empty or
contains one clause of the form $y=f(x_i)$ or one clause of the form
$f(y)=g(x_i)$ for some suitable $i$, $f$ and $g$; and $\Delta^{\neq}(\bar x y)$
contains arbitrarily many clauses of the form $y\neq f(x_i)$ or $f(y) \neq g(x_j)$.
Note that $\cw{\varphi}{\cGp}{\#}=\cw{\psi}{\cGpp}{\#}$, so it is enough
to compute the latter value.

Observe that it is enough to solve the weighted counting problem for each
disjunct separately, as we can then combine the results using a simple
inclusion-exclusion reasoning (the weighted sum for $q \lor q'$ is obtained by
adding the weighted sum for $q$ to the weighted sum for $q'$ and then
subtracting the weighted sum for $q \land q'$). In the sequel we then assume that $\psi$ has the form
described in~\eqref{eq-nf}.


The proof now goes by induction on the number of inequalities in
$\Delta^{\neq}$.  While the inductive step turns out to be fairly easy, the
difficult part is the base step of the induction.

We start with proving the inductive step.  Let $g(y) \neq f(x_i)$ be an
arbitrary inequality from $\Delta^{\neq}$ (where $g$ might possibly be the
identity).  Let $\psi^{-}$ be $\psi$ with this inequality removed and $\psi^{+}
= \psi^{-} \wedge g(y) = f(x_i)$. Of course $\psi$ and $\psi^{+}$ have disjoint
sets of solutions and we have:
\begin{equation*}
\cw{\psi}{\cGpp}{\#} = \cw{\psi^{-}}{\cGpp}{\#} - \cw{\psi^{+}}{\cGpp}{\#}.
\end{equation*}

Note that $\psi^-$ and $\psi^{+}$ have one less conjunct in
$\Delta^{\neq}$. The problem is that $\psi^+$ is not of the form~\eqref{eq-nf}
as it may now contain two elements in $\Delta^=$. However it can be seen that
the removal of the extra equality in $\Delta^=$ as described in the proof of
Proposition~\ref{prop-normalization} does not introduce any new elements in
$\Delta^{\neq}$.

\begin{clm}{\label{claim-ct-psiplus}.}
There exists a query $\psipp$ such that:
its size depends only on the size of $\psi^{+}$, $\psipp$
is in the normal form given by~\eqref{eq-nf},
it contains an inequality conjunct $h(y) \neq g_{1}(x_{i})$ (where
$h$ might possibly be identity) iff $\psi^{+}$ also contains
such conjunct and $\psipp(\cGpp) = \psi^{+}(\cGpp)$.
Moreover, $\psipp$ can be constructed in time linear in the
size of $\psi^{+}$.
\end{clm}

\begin{proof}
The proof is a simple case analysis of the content of $\Delta^{=}$ of $\psi$.

If its empty, then $\psipp$ is already in the desired form.

If it contains an atom of the form $y = h_{2}(x_j)$, then
equality $g(y) = f(x_i)$ is equivalent to $g(h_{2}(x_j)) = f(x_i)$
and we are done.

If it contains an atom of the form $h_{3}(y) = h_{2}(x_j)$
and $g$ is identity, then $h_{3}(y) = h_{2}(x_j)$ is equivalent
to $h_{3}(f(x_{i})) = h_{2}(x_j)$. If $g$ is not identity, then
$\tau(y)$ ensures us that either $g(y)$ determines $h_{3}(y)$
or vice versa. If we have $h_{4}(g(y)) = h_{3}(y)$, then
$h_{3}(y) = h_{2}(x_j)$ is equivalent to $h_{4}(f(x_i)) = h_{2}(x_j)$.
The other case is symmetric.

The fact that $\psipp$ does not contain any additional inequalities,
that it can be computed in time linear in the size of $\psi^{+}$
and that $\psipp(\cGpp) = \psi^{+}(\cGpp)$ follows from the
above construction.
\end{proof}

We can therefore remove
the extra element in $\Delta^+$ and assume that $\psi^+$ has the desired
form. We can now use the inductive hypothesis on the size of $\Delta^{\neq}$
to both $\psi^-$ and $\psi^+$ in order to compute both
$\cw{\psi^{-}}{\cGpp}{\#}$ and $\cw{\psi^+}{\cGpp}{\#}$ and derive
$\cw{\psi}{\cGpp}{\#}$.

It remains to show the base of the inner induction.
In the following we assume that $\Delta^{\neq}$ is empty.
The rest of the proof is a case analysis on the content of $\Delta^{=}$.

Assume first that $\Delta^{=}$ consists of an atom of the form $y = f(x_1)$.

Note that the solutions to $\psi$ are of the form $(\bar v f(v_1))$. We have:
{
\small
\begin{align*}
\cw{\psi}{\cGpp}{\#}\!
&= \!\!\!\!\!\!\!\!\!\sum_{(\bar v u) \in \psi(\cGpp)}
	\left(\#_k(u) \prod_{1 \leq i \leq k-1} \#_{i}(v_{i})\right)\\
&= \!\!\!\!\!\!\!\!\!\!\!\!\sum_{(\bar v f(v_1)) \in \psi(\cGpp)}
	\left(\#_k(f(v_1)) \prod_{1 \leq i \leq k-1} \#_{i}(v_{i}) \right)\\
&= \!\!\!\!\!\!\!\!\!\!\!\!\sum_{(\bar v f(v_1)) \in \psi(\cGpp)}
	\left( \#_1(v_1) \#_k(f(v_1))\!\!\! \prod_{2 \leq i \leq k-1}\!\!\! \#_{i}(v_{i})\right)
\end{align*}
}
In linear time we now iterate through all nodes $w$ in $\cGpp$ and set
\begin{align*}
\#'_1(w) &:= \#_1(w) \cdot \#_k(f(w))\\
\#'_{i}(w) &:= \#_{i}(w) & \text{ for } 2 \leq i \leq k-1.
\end{align*}
Let $\vartheta(\bar x)$ be $\psi$ with all occurrences of $y$ replaced with $f(x_1)$.
We then have:
{
\small
\begin{align*}
\cw{\psi}{\cGpp}{\#}
&= \sum_{(\bar v f(v_1)) \in \psi(\cGpp)}
	\left( \#'_1(v_1) \prod_{2 \leq i \leq k-1} \#'_{i}(v_{i})\right)\\
&= \sum_{\bar v \in \vartheta(\cGpp)}
	\prod_{1 \leq i \leq k-1} \#'_{i}(v_{i})\\
&= \cw{\vartheta}{\cGpp}{\#'}
\end{align*}
}

By induction on the number of free variables, as $\#'_i(w)$ can be computed in
constant time for each $i$ and $w$, we can compute $\cw{\vartheta}{\cGpp}{\#'}$ in time linear in
$\size{\cGpp}$ and we are done.

\medskip

Assume now that $\Delta^{=}$ consists of an atom $g(y) = f(x_1)$.  Let
$\psi'(y)$ be the formula $\exists \bar x \psi(\bar x y)$ and $\psi''(\bar x)$
the formula $\exists y \psi(\bar x y)$.  We first compute set $\psi'(\cGpp)$ in
linear time using Corollary~\ref{unary-FO}.  We now define a function $\#''_{k}
: V \rightarrow \Nat$ as:
\begin{equation*}
\#''_k(w) := \sum_{\substack{\set{u \in
    \psi'(\cGpp)\\ g(u) = w}}} \#_k(u).
\end{equation*}
Note that this function can be easily computed in linear time by going through
all nodes $w$ and adding $\#_k(w)$ to $\#''_k(g(w))$.

Finally we set:
\begin{align*}
\#'_1(w) &:= \#_1(w) \#''_k(f(w))\\
\#'_i(w) &:= \#_{i}(w) & \text{ for } 2 \leq i \leq k-1.
\end{align*}

Let $u_1, u_2 \in \psi'(\cGpp)$ be such that $g(u_1) = g(u_2)$.  Because
$\Delta^{\neq}$ is empty, observe that $\cGpp \models \forall \bar x
(\psi(\bar x u_1) \leftrightarrow \psi(\bar x u_2))$.  Based on this
observation we now group the solutions to $\psi$ according to their last $k-1$
values and get:

{
\small
\begin{align*}
\cw{\psi}{\cGpp}{\#}
&= \sum_{(\bar v u) \in \psi(\cGpp)}
	\left(\#_k(u) \prod_{1 \leq i \leq k-1} \#_{i}(v_{i})\right)\\
&=\!\!\! \sum_{\bar v \in \psi''(\cGpp)}
	\sum_{\substack{\set{u \in \psi'(\cGpp)\\ g(u) = f(v_1)}}}
	\left(\#_k(u) \prod_{1 \leq i \leq k-1} \#_{i}(v_{i}) \right)\\
&= \sum_{\bar v \in \psi''(\cGpp)} \left(\sum_{\substack{\set{u \in
        \psi'(\cGpp)\\ g(u) = f(v_1)}}} \#_k(u)\right)\prod_{1 \leq i \leq k-1} \#_{i}(v_{i})
	\\
&= \sum_{\bar v \in \psi''(\cGpp)}\left(
	\#''_k(f(v_1)) \prod_{1 \leq i \leq k-1} \#_{i}(v_{i})\right)\\
&=\!\!\!\! \sum_{\bar v \in \psi''(\cGpp)}
	\left(\#_1(v_1)\#''_k(f(v_1)) \prod_{2 \leq i \leq k-1} \#'_{i}(v_{i})\right)\\
&= \sum_{\bar v \in \psi''(\cGpp)}
	\prod_{1 \leq i \leq k-1} \#'_{i}(v_{i})\\
&= \cw{\psi''}{\cGpp}{\#'}
\end{align*}
}

By induction on the number of free variables, as $\#'_i(w)$ can be computed in
constant time for each $i$ and $w$, we can compute $\cw{\psi''}{\cGpp}{\#'}$
and we are done with this case.

The remaining case when $\Delta^{=}$ is empty is handled similarly to the previous one.
We then have
\begin{equation*}
\psi(\bar x y) = \psi_{1}(\bar x) \wedge \tau(y).
\end{equation*}

After setting
\begin{align*}
\#'_1(w) &:= \#_2(w) \cdot \sum_{u \in \tau(\cGpp)} \#_{1}(u)\\
\#'_i(w) &:= \#_{i+1}(w) & \text{ for } 2 \leq i \leq k-1
\end{align*}
we see that
\begin{equation*}
\cw{\psi}{\cGpp}{\#} = \cw{\psi_{1}}{\cGpp}{\#'}
\end{equation*}
and we conclude again by induction on the number of free variables.
\end{proof}
As we said earlier, Theorem~\ref{counting} is an immediate
consequence of Lemma~\ref{ct-lemma}.
\end{proof}

\section{Conclusions}

Queries written in first-order logic can be efficiently processed over the
class of structures having bounded expansion. We have seen that over this class
the problems investigated in this paper can be computed in time linear in the
size of the input structure. The constant factor however is high. The
approach taken here, as well as the ones of~\cite{DKT11,GK11},
yields a constant factor that is a tower of exponentials whose height depends
on the size of the query. This nonelementary constant factor is unavoidable
already on the class of unranked trees, assuming
FPT$\neq$AW$[*]$~\cite{FrickGrohe04}. In comparison, this factor can be
triply exponential in the size of the query in the bounded degree
case~\cite{Seese96,WL11}.

Since the submission of this work, the result has been extended to a larger
class of structures. In~\cite{NO11} the class of nowhere dense graphs was
introduced and it generalizes the notion of bounded expansion. It has
been shown that the model checking problem of first-order logic can be done in
nearly linear time (i.e., for any $\epsilon>0$ it can be done in
$O(n^{1+\epsilon})$) over any nowhere dense class of
graph~\cite{nowheredensemc}. Recently an enumeration procedure has
been proposed for first-order queries over nowhere dense graph classes, with a
nearly linear preprocessing time and constant
delay~\cite{DBLP:conf/pods/SchweikardtSV18}.

For graph classes closed under substructures, the nowhere dense property seems
to be the limit for having good algorithmic properties for first-order
logic. Indeed, it is known that the model checking problem of first-order logic
over a class of structures that is not nowhere dense cannot be FPT~\cite{KD09}
(modulo some complexity assumptions).

For structures of bounded expansion, an interesting open question is whether a
sampling of the solutions can be performed in linear time. For instance: can we
compute the $j$-th solution in constant time after a linear preprocessing? This
can be done in the bounded degree case~\cite{BDGO08} and in the bounded
treewidth case~\cite{Bagan09}. We leave the bounded expansion case for future
research.

Finally it would be interesting to know whether the index structure computed in
linear time for the enumeration process could be updated efficiently. In the
boolean case, queries of arity 0, updates can be done in constant
time~\cite{DKT11}, assuming the underlying graph is not changed too
much. In particular a relabeling of a node require only constant update
time. It would be interesting to know whether this constant update time could
also be achieved for an index structure allowing for constant delay
enumeration.

\bibliographystyle{alpha}
\bibliography{bibliography}

\end{document}